\documentclass[11pt]{extarticle}
\usepackage{extsizes}
\usepackage{amssymb,amsfonts,amsthm}
\usepackage{amsmath}
\usepackage[makeroom]{cancel}
\usepackage{mathtools}
\usepackage{mathrsfs,pifont}
\usepackage{slashed,mathabx} 
\usepackage[bbgreekl]{mathbbol}
\usepackage{ytableau} 
\usepackage{tikz}
\usetikzlibrary{calc}

\usepackage[top=1in, bottom=1.25in, left=0.75in,right=0.75in]{geometry}
\usepackage{authblk}

\usepackage[font={small,sl}]{caption}
\usepackage[all]{xy}

\usepackage{hyperref}

\newtheorem{theorem}{Theorem}[section]
\newtheorem{lemma}[theorem]{Lemma}
\newtheorem{corollary}[theorem]{Corollary}
\theoremstyle{definition}
\newtheorem{definition}[theorem]{Definition}

\newtheorem{proposition}[theorem]{Proposition}
\newtheorem{example}[theorem]{Example}

\theoremstyle{remark}

\numberwithin{equation}{section}

\makeatletter
\def\blfootnote{\xdef\@thefnmark{}\@footnotetext}
\makeatother

\newcommand{\Aq}{\mathcal A_q}
\newcommand{\Oq}{O_q}
\newcommand{\cB}{\mathcal B}
\newcommand{\cW}{\mathcal W}
\newcommand{\cG}{\mathcal G}
\newcommand{\Gt}{\widetilde{\cG}}

\newcommand{\F}{\mathbb F}
\newcommand{\N}{\mathbb N}

\DeclareMathOperator{\Span}{Span}

\begin{document}

\title{A new PBW basis for the alternating central extension of \\the $q$-Onsager algebra}

\author{Haoran Zhu
\thanks{Division of Mathematical Sciences, Nanyang Technological University,
21 Nanyang Link, Singapore 637371. Email:
\texttt{zhuh0031@e.ntu.edu.sg}}}

\date{}
\maketitle

\begin{abstract}
We establish a new PBW basis for $\Aq$, the alternating central extension of the
$q$-Onsager algebra. Terwilliger showed that the alternating generators form
a PBW basis in the block order
$\cG<\cW^-<\cW^+<\widetilde{\cG}$. We prove that they also form a PBW basis
in the different block order
$\cW^-<\cG<\widetilde{\cG}<\cW^+$. Consequently, multiplication induces a
vector-space isomorphism
\[
\cW^-\otimes\cG\otimes\widetilde{\cG}\otimes\cW^+
\longrightarrow \Aq,
\]
thereby confirming a conjecture of Terwilliger.
\end{abstract}


\medskip


\section{Introduction}

The Onsager algebra is one of the classical algebraic structures behind
exactly solvable lattice models.
It entered mathematical physics through Onsager's solution of the
two-dimensional Ising model \cite{Onsager1944}, and the Dolan--Grady
relations \cite{DolanGrady1982} later gave a concise algebraic mechanism for
producing the infinite families of commuting operators that underlie
Onsager-type integrability.
The \textbf{$q$-Onsager algebra} $\Oq$ is a $q$-deformation of this
structure.
It appears naturally in boundary quantum integrable systems
\cite{Baseilhac2005Integrable,Baseilhac2005Deformed,
BaseilhacKoizumi2005New,BaseilhacKoizumi2005Symmetry,
BaseilhacBelliard2010Generalized,
LemartheBaseilhacGainutdinov2026Fused,
BaseilhacGainutdinovLemarthe2025TT},
in the theory of tridiagonal pairs and tridiagonal algebras, in algebraic
combinatorics and representation theory
\cite{Terwilliger1993Subconstituent,Terwilliger2001,
ItoTanabeTerwilliger2001,ItoTerwilliger2009,ItoTerwilliger2010,
ItoNomuraTerwilliger2011,Terwilliger2018UAW,
Terwilliger2024S3},
and as a basic example of a quantum symmetric pair coideal subalgebra
\cite{Letzter2002,Kolb2014,Terwilliger2018Lusztig,
BaseilhacKolb2020,LuWang2021Drinfeld,LuRuanWang2023Hall}.

The present paper concerns the algebra $\Aq$ associated with the
$q$-Onsager algebra.
The algebraic structure underlying $\Aq$ appeared in the work of Baseilhac
and Koizumi \cite{BaseilhacKoizumi2005New}.
Baseilhac and Shigechi subsequently introduced the current algebra $\Aq$ in
the setting of Sklyanin's reflection algebra and gave its presentation in
terms of four infinite families of generators
\cite{Sklyanin1988,BaseilhacShigechi2010}.
On the basis of supporting evidence, $\Aq$ was conjectured to be isomorphic
to a central extension of $\Oq$
\cite{BaseilhacBelliard2017Attractive,Terwilliger2021Conjecture}.
Terwilliger proved this conjecture and identified $\Aq$ with the
\textbf{alternating central extension} of the $q$-Onsager algebra
\cite{Terwilliger2021ACE}; see also
\cite{Terwilliger2021OqACE,Terwilliger2023Compact}.
We use this terminology throughout the paper.
This alternating central extension is different from the central extension
of $\Oq$ appearing in the Drinfeld-type presentation of affine
$\imath$-quantum groups \cite{LuWang2021Drinfeld}.

The four families in the presentation of $\Aq$ are known as the
\textbf{alternating generators}; following Terwilliger's notation, we write
them as
\[
        \{\cW_{-k}\}_{k\in\N},\qquad
        \{\cW_{k+1}\}_{k\in\N},\qquad
        \{\cG_{k+1}\}_{k\in\N},\qquad
        \{\Gt_{k+1}\}_{k\in\N}.
\]
We reserve the corresponding straight symbols
\[
        W_{-k},\qquad W_{k+1},\qquad
        G_{k+1},\qquad \widetilde G_{k+1}
\]
for their images in $\Oq$ under a central reduction.
The presentation has many defining relations, but its four-family form is
highly structured.
This structure suggests that $\Aq$ should admit PBW bases and tensor
factorisations analogous to triangular decompositions in Lie-theoretic
settings.

This expectation is supported by several parallel developments around the
$q$-Onsager algebra.
Baseilhac and Kolb \cite{BaseilhacKolb2020} constructed root vectors and a
PBW basis for the $q$-Onsager algebra itself using Lusztig-type
automorphisms.
Terwilliger studied the $q$-Onsager algebra, its alternating central
extension, the alternating generators, and related structures in a series
of papers
\cite{Terwilliger2017PositivePart,Terwilliger2018UAW,
Terwilliger2018Action,Terwilliger2021Conjecture,
Terwilliger2021ACE,Terwilliger2021OqACE,
Terwilliger2023Compact,Terwilliger2024S3}.
In another direction, Drinfeld-type and Hall-algebra approaches place the
$q$-Onsager algebra within the broader theory of affine
$\imath$-quantum groups
\cite{LuWang2021Drinfeld,LuRuanWang2023Hall}.
More recently, the alternating central extension has appeared in universal
$K$-operator and transfer-matrix constructions
\cite{LemartheBaseilhacGainutdinov2026Fused,
BaseilhacGainutdinovLemarthe2025TT}.

In \cite[Conjecture~4.5]{Terwilliger2018Action}, Terwilliger formulated a
conjectural PBW basis for $\Aq$.
The conjecture consists of two assertions.
First, each of the four alternating generator families should freely
generate a polynomial subalgebra.
Second, the multiplication map from the tensor product of these four
polynomial subalgebras to $\Aq$ should be a vector-space isomorphism.
More explicitly, set
\begin{align*}
\cW^-&=\langle \cW_0,\cW_{-1},\cW_{-2},\ldots\rangle,
&
\cG&=\langle \cG_1,\cG_2,\cG_3,\ldots\rangle,\\
\cW^+&=\langle \cW_1,\cW_2,\cW_3,\ldots\rangle,
&
\widetilde{\cG}
&=\langle \Gt_1,\Gt_2,\Gt_3,\ldots\rangle,
\end{align*}
where $\langle\cdot\rangle$ denotes the subalgebra generated by the
indicated elements.
The conjectural tensor factorisation is
\[
        \cW^-\otimes\cG\otimes
        \widetilde{\cG}\otimes\cW^+
        \longrightarrow\Aq,
        \qquad
        u\otimes v\otimes w\otimes x\longmapsto uvwx.
\]
Thus, the block order in the conjecture is
\[
        \cW^-<\cG<\widetilde{\cG}<\cW^+.
\]
The corresponding ordered monomials are expected to form the desired PBW
basis of $\Aq$.

A major step towards this conjecture was taken by Terwilliger in
\cite[Theorem~6.1]{Terwilliger2021ACE}.
There, he proved that the alternating generators give a PBW basis for
$\Aq$ in any linear order satisfying
\[
        \cG_{i+1}<\cW_{-j}<\cW_{k+1}<\Gt_{\ell+1}
        \qquad (i,j,k,\ell\in\N).
\]
The corresponding block order is
\[
        \cG<\cW^-<\cW^+<\widetilde{\cG}.
\]
Using this PBW basis, Terwilliger also proved the tensor product
factorisation
\[
        \Aq\simeq\Oq\otimes\F[z_1,z_2,\ldots],
\]
where $\F$ is the ground field and the $z_i$ are mutually commuting indeterminates \cite[Theorems~10.2--10.4]{Terwilliger2021ACE}.
In particular, this PBW theorem immediately implies that each of the four
alternating generator families freely generates a polynomial subalgebra.
We record this consequence in Section~3 for completeness.
The remaining problem is to prove that the alternating generators also give
a PBW basis in the different block order
\[
        \cW^-<\cG<\widetilde{\cG}<\cW^+.
\]

The conjectural block order has recently appeared in computations of
universal TT- and TQ-relations for the alternating central extension of the
$q$-Onsager algebra
\cite{BaseilhacGainutdinovLemarthe2025TT}.
In that work, this order is used as a PBW order for calculations involving
the alternating generators.
Terwilliger's PBW theorem establishes a PBW basis in the different block
order
\[
        \cG<\cW^-<\cW^+<\widetilde{\cG},
\]
and therefore does not by itself justify the conjectural order. A separate
proof is required.

The choice of PBW order can also be important in related constructions.
For quantum affine algebras, factorised formulae for universal
$R$-matrices are associated with specific PBW orders; see, for example,
\cite[Theorem~7.1]{KhoroshkinTolstoy1991} and
\cite{Damiani1998}.
Likewise, the associated $L$-operators can admit LDU decompositions adapted
to a particular order; see, for example, the entries displayed in
\cite[eqs.~(4.8) and~(4.9)]{DingFrenkel1993}.
By analogy, in the setting of coideal or comodule algebras, the new PBW
order may be useful in the study of universal $K$-matrices and related
$K$-operators
\cite{AppelVlaar2022,LemartheBaseilhacGainutdinov2026Fused,
BaseilhacGainutdinovLemarthe2025TT}.
Moreover, after a central reduction from $\Aq$ to $\Oq$, it would be
interesting to investigate whether the crossing analysis developed here can
be adapted to the alternating-generator PBW conjectures in
\cite[Conjecture~1]{BaseilhacBelliard2017Attractive} and
\cite[Conjecture~16.2]{Terwilliger2021OqACE}.

The following theorem records the full statement of Terwilliger's conjecture \cite[Conjecture~4.5]{Terwilliger2018Action}.
Parts~{\rm (i)--(iv)} follow directly from Terwilliger's PBW theorem \cite[Theorem~6.1]{Terwilliger2021ACE} and are included for completeness.
The principal new result is the PBW basis in the conjectural block order; together with parts~{\rm (i)--(iv)}, it yields part~{\rm (v)}.

\begin{theorem}\label{thm:intro-main}
Assume that $0\neq q\in\F$ is not a root of unity. Then:
\begin{enumerate}
\item[{\rm (i)}] $\cW^-$ is a polynomial algebra freely generated by
        $\cW_0,\cW_{-1},\cW_{-2},\ldots$.
\item[{\rm (ii)}] $\cG$ is a polynomial algebra freely generated by
        $\cG_1,\cG_2,\cG_3,\ldots$.
\item[{\rm (iii)}] $\widetilde{\cG}$ is a polynomial algebra freely
        generated by $\Gt_1,\Gt_2,\Gt_3,\ldots$.
\item[{\rm (iv)}] $\cW^+$ is a polynomial algebra freely generated by
        $\cW_1,\cW_2,\cW_3,\ldots$.
\item[{\rm (v)}] The multiplication map
\[
        \cW^-\otimes\cG\otimes
        \widetilde{\cG}\otimes\cW^+
        \longrightarrow\Aq,\qquad
        u\otimes v\otimes w\otimes x\longmapsto uvwx
\]
is an isomorphism of vector spaces.
\end{enumerate}
\end{theorem}

The proof of part~{\rm (v)} proceeds by showing that the alternating
generators give a PBW basis in the conjectural block order
\[
        \cW^-<\cG<\widetilde{\cG}<\cW^+.
\]

The paper is organised as follows.
In Section~2 we recall the alternating central extension $\Aq$, the
alternating generators, the degree function, and Terwilliger's PBW theorem.
In Section~3 we record, for completeness, that the four single-family
subalgebras are polynomial algebras.
In Section~4 we establish the two finite crossing computations and the
determinant calculation.
In Section~5 we use these computations to change the PBW order and prove the
conjectural PBW basis.
In Section~6 we deduce the tensor factorisation and complete the proof of
Theorem~\ref{thm:intro-main}.

\section{The alternating central extension of the \texorpdfstring{$q$}{q}-Onsager algebra}

We recall the current-algebra presentation of the alternating central
extension $\Aq$ and fix the notation used throughout the paper. 
Let $\F$ denote a field, and let $\N=\{0,1,2,\ldots\}$.
All vector spaces and tensor products are over $\F$, and all algebras are associative unital $\F$-algebras.
Throughout, $q$ denotes a non-zero element of $\F$ which is not a root of unity. 
For elements $X,Y$ of an algebra, we write
\[
        [X,Y]=XY-YX,
        \qquad
        [X,Y]_q=qXY-q^{-1}YX.
\]
We also set
\[
        [2]_q=q+q^{-1}.
\]

The following presentation is due to Baseilhac and Shigechi~\cite[Definition~3.1]{BaseilhacShigechi2010}.  
We follow the notation of Terwilliger \cite[Definition~2.4]{Terwilliger2021ACE}.

\begin{definition}\label{def:Aq}
The alternating central extension $\Aq$ of the $q$-Onsager algebra is the
$\F$-algebra generated by the four families
\[
        \cW_{-k},\qquad \cW_{k+1},\qquad
        \cG_{k+1},\qquad \Gt_{k+1}
        \qquad (k\in\N).
\]
These generators are called the \textbf{alternating generators} of $\Aq$.
The defining relations are, for all $k,\ell\in\N$:
\begin{align}
&[\cW_0,\cW_{k+1}]
        =[\cW_{-k},\cW_1]
          =\frac{\Gt_{k+1}-\cG_{k+1}}{q+q^{-1}},
          \label{eq:Aq1}\\
&[\cW_0,\cG_{k+1}]_q
        =[\Gt_{k+1},\cW_0]_q
          =\rho \cW_{-k-1}-\rho \cW_{k+1},
          \label{eq:Aq2}\\
&[\cG_{k+1},\cW_1]_q
        =[\cW_1,\Gt_{k+1}]_q
          =\rho \cW_{k+2}-\rho \cW_{-k},
          \label{eq:Aq3}\\
&[\cW_{-k},\cW_{-\ell}]=0,
        \qquad
        [\cW_{k+1},\cW_{\ell+1}]=0,
        \label{eq:Aq4}\\
&[\cW_{-k},\cW_{\ell+1}]
        +[\cW_{k+1},\cW_{-\ell}]=0,
        \label{eq:Aq5}\\
&[\cW_{-k},\cG_{\ell+1}]
        +[\cG_{k+1},\cW_{-\ell}]=0,
        \label{eq:Aq6}\\
&[\cW_{-k},\Gt_{\ell+1}]
        +[\Gt_{k+1},\cW_{-\ell}]=0,
        \label{eq:Aq7}\\
&[\cW_{k+1},\cG_{\ell+1}]
        +[\cG_{k+1},\cW_{\ell+1}]=0,
        \label{eq:Aq8}\\
&[\cW_{k+1},\Gt_{\ell+1}]
        +[\Gt_{k+1},\cW_{\ell+1}]=0,
        \label{eq:Aq9}\\
&[\cG_{k+1},\cG_{\ell+1}]=0,
        \qquad
        [\Gt_{k+1},\Gt_{\ell+1}]=0,
        \label{eq:Aq10}\\
&[\Gt_{k+1},\cG_{\ell+1}]
        +[\cG_{k+1},\Gt_{\ell+1}]=0,
        \label{eq:Aq11}
\end{align}
where $\rho=-(q^2-q^{-2})^2$.
\end{definition}

Following \cite[Definition~2.4, Eq.~(14)]{Terwilliger2021ACE}, we use the convention
\begin{equation}\label{eq:G0}
        \cG_0=\Gt_0=-(q-q^{-1})[2]_q^2.
\end{equation}
Thus $\cG_0$ and $\Gt_0$ are scalars, not alternating generators.
This convention will be used in the reduction rules, where terms of index zero may occur.

We next fix the PBW convention used in the sequel.

\begin{definition}
\label{def:PBW}
Let $A$ be an $\F$-algebra, let $\Omega\subseteq A$, and let $<$ be a linear
order on $\Omega$.  We say that $\Omega$, with the order $<$, gives a
\textbf{PBW basis} for $A$ if the ordered monomials
\[
        a_1a_2\cdots a_n,
        \qquad
        n\geq 0,\quad
        a_1,\ldots,a_n\in\Omega,\quad
        a_1\leq a_2\leq\cdots\leq a_n,
\]
form a basis of $A$ over $\F$.  For $n=0$, the monomial is interpreted as the
identity element of $A$.
\end{definition}

The following theorem of Terwilliger is the starting point for our argument.

\begin{theorem}[{\cite[Theorem~6.1]{Terwilliger2021ACE}}]
\label{thm:terwilliger-pbw}
Assume that $0\neq q\in\F$ is not a root of unity.  Then the alternating
generators of $\Aq$, equipped with any linear order $<$ satisfying
\begin{equation}
\label{eq:terw-order}
        \cG_{i+1}<\cW_{-j}<\cW_{k+1}<\Gt_{\ell+1}
        \qquad (i,j,k,\ell\in\N),
\end{equation}
give a PBW basis for $\Aq$.
\end{theorem}

For the rest of the paper, we fix one such order and denote it by
$<_{\rm T}$.
Within each of the four families, we choose the order by increasing index:
\[
        \cG_{i+1}<_{\rm T}\cG_{j+1},\qquad
        \cW_{-i}<_{\rm T}\cW_{-j},\qquad
        \cW_{i+1}<_{\rm T}\cW_{j+1},\qquad
        \Gt_{i+1}<_{\rm T}\Gt_{j+1}
        \qquad (i<j).
\]
Thus
\[
        \cG_1<_{\rm T}\cG_2<_{\rm T}\cdots
        <_{\rm T}\cW_0<_{\rm T}\cW_{-1}<_{\rm T}\cW_{-2}<_{\rm T}\cdots
        <_{\rm T}\cW_1<_{\rm T}\cW_2<_{\rm T}\cW_3<_{\rm T}\cdots
        <_{\rm T}\Gt_1<_{\rm T}\Gt_2<_{\rm T}\cdots .
\]
Let $\cB_{\rm T}$ denote the corresponding PBW basis. 
This is the known PBW basis from which we shall change to the conjectural order.

We shall also use the degree function introduced in
\cite[Definition~7.1]{Terwilliger2021ACE}.
On the alternating generators, it is given by
\begin{equation}
\label{eq:degree}
        \deg \cW_{-k}=\deg \cW_{k+1}=2k+1,
        \qquad
        \deg \cG_{k+1}=\deg \Gt_{k+1}=2k+2
        \qquad (k\in\N).
\end{equation}
For a word $w$ in the alternating generators, $\deg w$ is the sum of the
degrees of its letters. We write $\operatorname{len}(w)$ for the length of
$w$, namely the number of alternating generators appearing in it, and define
\[
        \operatorname{bideg}(w)
        =\bigl(\deg w,\operatorname{len}(w)\bigr).
\]
We order bidegrees lexicographically, comparing degree first. The scalars
$\cG_0$ and $\Gt_0$ from \eqref{eq:G0} are regarded as having degree zero
and length zero.

The fixed PBW basis $\cB_{\rm T}$ gives a convenient degree--length decomposition.  
For $d,r\in\N$, set
\[
        {\Aq}_{d,r}
        =
        \Span\{b\in\cB_{\rm T}\mid
        \deg b=d,\; \operatorname{len}(b)=r\}.
\]
Then
\[
        \Aq=\bigoplus_{d,r\in\N}{\Aq}_{d,r}
\]
as a vector space. 

For the filtration arguments below, define
\begin{equation}
\label{eq:lower-space}
        \Aq^{\lhd(d,r)}
        =
        \sum_{0\leq e<d}\sum_{s\geq 0}{\Aq}_{e,s}
        +
        \sum_{0\leq s<r}{\Aq}_{d,s}.
\end{equation}
Thus $\Aq^{\lhd(d,r)}$ is the span of the PBW basis elements whose bidegree is
lower than $(d,r)$ in lexicographic order, with degree compared first.  Since
every alternating generator has positive degree, only finitely many lengths
occur in any fixed degree.

Finally, set
\[
        B_d=\sum_{e=0}^{d}\sum_{r\geq 0}{\Aq}_{e,r}
        \qquad (d\in\N).
\]
This is the degree filtration piece determined by $\cB_{\rm T}$.  Terwilliger's
Hilbert-series computation \cite[Lemma~7.6]{Terwilliger2021ACE} gives
\begin{equation}
\label{eq:hilbert}
        \sum_{d\geq 0}
        \dim\left(\sum_{r\geq 0}{\Aq}_{d,r}\right)x^d
        =
        \prod_{n=1}^{\infty}\frac{1}{(1-x^n)^2}.
\end{equation}
Equivalently, the graded dimension is that of a polynomial algebra with two
generators in every positive degree.  This dimension count will be used in
Section~5, after the spanning argument, to prove linear independence.

\section{The four polynomial subalgebras}

For completeness, we record the polynomiality assertions in
\cite[Conjecture~4.5(i)--(iv)]{Terwilliger2018Action}.
They are immediate consequences of Terwilliger's PBW theorem; see
Theorem~\ref{thm:terwilliger-pbw} and
\cite[below Theorem~6.1]{Terwilliger2021ACE}.

For a subset $S\subseteq\Aq$, write $\langle S\rangle$ for the unital
subalgebra of $\Aq$ generated by $S$. Define
\[
        \cW^-=
        \langle \cW_0,\cW_{-1},\cW_{-2},\ldots\rangle,
        \qquad
        \cG=
        \langle \cG_1,\cG_2,\cG_3,\ldots\rangle,
\]
and
\[
        \widetilde{\cG}
        =
        \langle \Gt_1,\Gt_2,\Gt_3,\ldots\rangle,
        \qquad
        \cW^+
        =
        \langle \cW_1,\cW_2,\cW_3,\ldots\rangle.
\]

\begin{proposition}\label{prop:poly}
Let $\F[\lambda_1,\lambda_2,\ldots]$ denote the polynomial algebra in
countably many commuting indeterminates. Each of the following assignments
extends to an algebra isomorphism:
\[
\begin{aligned}
\F[\lambda_1,\lambda_2,\ldots]
    &\longrightarrow \cW^-,
&\qquad \lambda_{k+1}
    &\longmapsto \cW_{-k},\\
\F[\lambda_1,\lambda_2,\ldots]
    &\longrightarrow \cG,
&\qquad \lambda_{k+1}
    &\longmapsto \cG_{k+1},\\
\F[\lambda_1,\lambda_2,\ldots]
    &\longrightarrow \widetilde{\cG},
&\qquad \lambda_{k+1}
    &\longmapsto \Gt_{k+1},\\
\F[\lambda_1,\lambda_2,\ldots]
    &\longrightarrow \cW^+,
&\qquad \lambda_{k+1}
    &\longmapsto \cW_{k+1},
\end{aligned}
\qquad (k\in\N).
\]
Consequently, each of the four displayed subalgebras is a polynomial
algebra freely generated by its indicated family of alternating generators.
\end{proposition}

\begin{proof}
By \eqref{eq:Aq4} and \eqref{eq:Aq10}, the generators within each of the
four families commute pairwise. Hence each displayed assignment extends,
by the universal property of the polynomial algebra, to a surjective
algebra homomorphism.

The standard monomial basis of
$\F[\lambda_1,\lambda_2,\ldots]$ is mapped to the corresponding
single-family ordered monomials. With the increasing-index order fixed in
Section~2, these monomials are elements of the PBW basis
$\cB_{\rm T}$ from Theorem~\ref{thm:terwilliger-pbw}. They are therefore
linearly independent. Each of the four homomorphisms is consequently
injective, and hence is an algebra isomorphism.
\end{proof}

\section{The two crossing computations}

This section contains the explicit calculations needed in the proof.
We compare the known PBW block order $\cG<\cW^-<\cW^+<\Gt$ with the conjectural block order $\cW^-<\cG<\Gt<\cW^+$.
Thus, it is enough to control the two adjacent block crossings
\[
        \cG\cW^- \leftrightarrow \cW^-\cG,
        \qquad
        \cW^+\Gt \leftrightarrow \Gt\cW^+.
\]
Terwilliger's reduction rules
\cite[Section~5]{Terwilliger2021ACE} show that, modulo terms of lower
bidegree, each crossing is governed by a finite matrix.
We begin with this matrix.

Let $N\in\N$.  
For $0\leq a\leq N$, set
\[
        m_a=\min(a,N-a).
\]
For a parameter $\eta$, let $M_N(\eta)$ be the $(N+1)\times (N+1)$ matrix whose rows and columns are indexed by $0,1,\ldots,N$, and whose row indexed by $a$ is defined by
\begin{equation}\label{eq:MN-row}
        \sum_{b=0}^{N} (M_N(\eta))_{ab} Z_b
        =
        Z_a+\eta\sum_{b=0}^{m_a-1}Z_b
        -\eta\sum_{b=N-m_a}^{N}Z_b
\end{equation}
for indeterminates $Z_0,Z_1,\ldots,Z_N$. 
Empty sums are interpreted as zero.
In particular, when $m_a=0$, the final sum in \eqref{eq:MN-row} consists of the single term $Z_N$.

\begin{example}
Let $N=4$. Then $m_0=0$, $m_1=1$, $m_2=2$, $m_3=1$, and $m_4=0$.
With respect to the ordered basis $Z_0,Z_1,Z_2,Z_3,Z_4$, the matrix $M_4(\eta)$ is given by
\[
M_4(\eta)=
\begin{pmatrix}
1 & 0 & 0 & 0 & -\eta\\
\eta & 1 & 0 & -\eta & -\eta\\
\eta & \eta & 1-\eta & -\eta & -\eta\\
\eta & 0 & 0 & 1-\eta & -\eta\\
0 & 0 & 0 & 0 & 1-\eta
\end{pmatrix}.
\]
For instance, the row indexed by \(a=2\) is obtained from $Z_2+\eta(Z_0+Z_1)-\eta(Z_2+Z_3+Z_4)$, which explains the entry \(1-\eta\) in the third diagonal position.
\end{example}

\begin{lemma}\label{lem:detMN}
For $N\in\N$,
\[
        \det M_N(\eta)=(1-\eta)^{\lfloor N/2\rfloor+1}.
\]
\end{lemma}

\begin{proof}
For $N=0$ and $N=1$, one has $ M_0(\eta)=(1-\eta)$ and $M_1(\eta)= \left( \begin{smallmatrix} 1 & -\eta\\0 & 1-\eta \end{smallmatrix} \right)$, so the formula is immediate.

Assume $N\geq 2$.  Since $m_N=0$, the last row of $M_N(\eta)$ has only one non-zero entry, namely $1-\eta$ in the last column. 
Expanding along this row gives a factor $1-\eta$.  After deleting the last row and the last column, the first row of the remaining matrix has only one non-zero entry, namely $1$ in the first column.  
Expanding along this row and relabelling the remaining rows and columns gives $M_{N-2}(\eta)$. 
Hence
\[
\det M_N(\eta)=(1-\eta)\det M_{N-2}(\eta).
\]
The result follows by induction on $N$.
\end{proof}

\begin{lemma}[{\cite[Proposition~5.1(v), (vi)]{Terwilliger2021ACE}}]
\label{lem:leading-reductions}
Let $i,j\in\N$, and set $m=\min(i,j)$. Modulo
$\Aq^{\lhd(2i+2j+3,2)}$, one has
\begin{align}
\cW_{-i}\cG_{j+1}
&\equiv
\cG_{j+1}\cW_{-i}
+(1-q^{-2})\sum_{\ell=1}^{m}
        \cG_\ell\cW_{-(i+j+1-\ell)}
-(1-q^{-2})\sum_{\ell=0}^{m}
        \cG_{i+j+1-\ell}\cW_{-\ell},
\label{eq:leading-Wminus-G}\\
\Gt_{i+1}\cW_{j+1}
&\equiv
\cW_{j+1}\Gt_{i+1}
+(1-q^2)\sum_{\ell=1}^{m}
        \cW_{i+j+2-\ell}\Gt_\ell
-(1-q^2)\sum_{\ell=0}^{m}
        \cW_{\ell+1}\Gt_{i+j+1-\ell}.
\label{eq:leading-tildeG-Wplus}
\end{align}
\end{lemma}

\paragraph{The crossing
\(\cG\cW^-\leftrightarrow\cW^-\cG\)}

The first crossing is obtained by putting $i=N-a$ and $j=a$ in
Lemma~\ref{lem:leading-reductions}.

\begin{proposition}\label{prop:first-crossing}
Let $N\in\N$. For $0\leq a\leq N$,
\begin{align}
\cW_{-(N-a)}\cG_{a+1}
\equiv{}&
\cG_{a+1}\cW_{-(N-a)}
+(1-q^{-2})\sum_{b=0}^{m_a-1}
        \cG_{b+1}\cW_{-(N-b)}
\notag\\
&-(1-q^{-2})\sum_{b=N-m_a}^{N}
        \cG_{b+1}\cW_{-(N-b)}
        \pmod{\Aq^{\lhd(2N+3,2)}}.
\label{eq:first-crossing}
\end{align}
Equivalently,
\[
\begin{pmatrix}
        \cW_{-N}\cG_1\\
        \cW_{-(N-1)}\cG_2\\
        \vdots\\
        \cW_0\cG_{N+1}
\end{pmatrix}
\equiv
M_N(1-q^{-2})
\begin{pmatrix}
        \cG_1\cW_{-N}\\
        \cG_2\cW_{-(N-1)}\\
        \vdots\\
        \cG_{N+1}\cW_0
\end{pmatrix}
\pmod{\Aq^{\lhd(2N+3,2)}}.
\]
\end{proposition}

\paragraph{The crossing
\(\cW^+\Gt\leftrightarrow\Gt\cW^+\)}

The second crossing is obtained by putting $i=a$ and $j=N-a$ in
Lemma~\ref{lem:leading-reductions}.

\begin{proposition}
\label{prop:second-crossing}
Let $N\in\N$. For $0\leq a\leq N$,
\begin{align}
\Gt_{a+1}\cW_{N-a+1}
\equiv{}&
\cW_{N-a+1}\Gt_{a+1}
+(1-q^2)\sum_{b=0}^{m_a-1}
        \cW_{N-b+1}\Gt_{b+1}
\notag\\
&-(1-q^2)\sum_{b=N-m_a}^{N}
        \cW_{N-b+1}\Gt_{b+1}
        \pmod{\Aq^{\lhd(2N+3,2)}}.
\label{eq:second-crossing}
\end{align}
Equivalently,
\[
\begin{pmatrix}
        \Gt_1\cW_{N+1}\\
        \Gt_2\cW_N\\
        \vdots\\
        \Gt_{N+1}\cW_1
\end{pmatrix}
\equiv
M_N(1-q^2)
\begin{pmatrix}
        \cW_{N+1}\Gt_1\\
        \cW_N\Gt_2\\
        \vdots\\
        \cW_1\Gt_{N+1}
\end{pmatrix}
\pmod{\Aq^{\lhd(2N+3,2)}}.
\]
\end{proposition}

\begin{corollary}
\label{cor:invert-crossings}
The two crossings are invertible modulo lower bidegree. More precisely, for
each $N\in\N$ and $0\leq b\leq N$, there exist scalars
$c_{ba}^{(N)}$ and $d_{ba}^{(N)}$, indexed by $0\leq a\leq N$, such that
\begin{equation}
\label{eq:inverse-crossing-one}
        \cG_{b+1}\cW_{-(N-b)}
        \equiv
        \sum_{a=0}^{N}c_{ba}^{(N)}
        \cW_{-(N-a)}\cG_{a+1}
        \pmod{\Aq^{\lhd(2N+3,2)}},
\end{equation}
and
\begin{equation}
\label{eq:inverse-crossing-two}
        \cW_{N-b+1}\Gt_{b+1}
        \equiv
        \sum_{a=0}^{N}d_{ba}^{(N)}
        \Gt_{a+1}\cW_{N-a+1}
        \pmod{\Aq^{\lhd(2N+3,2)}}.
\end{equation}
\end{corollary}

\begin{proof}
By Lemma~\ref{lem:detMN},
\[
        \det M_N(1-q^{-2})
        =
        q^{-2(\lfloor N/2\rfloor+1)}
        \neq 0,
\]
and
\[
        \det M_N(1-q^2)
        =
        q^{2(\lfloor N/2\rfloor+1)}
        \neq 0.
\]
Thus the two matrices appearing in
Propositions~\ref{prop:first-crossing}
and~\ref{prop:second-crossing} are invertible. Taking their inverse matrices
gives \eqref{eq:inverse-crossing-one} and
\eqref{eq:inverse-crossing-two}.
\end{proof}

\section{A new PBW basis for the alternating central extension}

We now pass from the two crossing computations to the PBW order conjectured
in \cite[Conjecture~4.5]{Terwilliger2018Action}.
Let $\cB_{\rm C}$ denote the set of all monomials
\begin{equation}\label{eq:conjectural-basis-monomial}
        \cW_{-i_1}\cdots \cW_{-i_r}\,
        \cG_{j_1+1}\cdots \cG_{j_s+1}\,
        \Gt_{k_1+1}\cdots \Gt_{k_t+1}\,
        \cW_{\ell_1+1}\cdots \cW_{\ell_u+1},
\end{equation}
where $r,s,t,u\in\N$ and
\[
        0\leq i_1\leq\cdots\leq i_r,\qquad
        0\leq j_1\leq\cdots\leq j_s,
\]
\[
        0\leq k_1\leq\cdots\leq k_t,\qquad
        0\leq \ell_1\leq\cdots\leq \ell_u.
\]
Empty blocks are allowed. Thus the block order is
\[
        \cW^-<\cG<\widetilde{\cG}<\cW^+.
\]
Set $V_{\rm C}=\Span(\cB_{\rm C})$.
For $d\in\N$, let $V_{\rm C}^{\leq d}$ denote the span of those monomials
in $\cB_{\rm C}$ whose degree is at most $d$.
We shall prove that $\cB_{\rm C}$ is a PBW basis for $\Aq$.

We first record two elementary facts about terms of lower bidegree.
By an $\Aq$-word we mean a word in the alternating generators.

\begin{lemma}
\label{lem:straightening-respects-bidegree}
Let $w$ be an $\Aq$-word of degree $d$ and length $r$.  When $w$ is expanded
in the PBW basis $\cB_{\rm T}$, only basis elements of degree less than $d$,
or of degree $d$ and length at most $r$, can occur.
\end{lemma}

\begin{proof}
The straightening procedure for the PBW basis $\cB_{\rm T}$ uses the
reduction rules in \cite[Proposition~5.1]{Terwilliger2021ACE}.  Each
reduction replaces a reducible word of length two by a linear combination of
words of length at most two, and so length does not increase.  Moreover,
degree does not increase under these reductions by
\cite[Lemma~7.11]{Terwilliger2021ACE}.  Repeatedly straightening $w$ therefore
gives only PBW basis elements of degree at most $d$, and among those of degree
$d$ only basis elements of length at most $r$ can occur.
\end{proof}

\begin{lemma}\label{lem:lower-stable}
Let $R\in \Aq^{\lhd(d,r)}$. 
Let $P,Q$ be $\Aq$-words of degrees $d_P,d_Q$ and lengths $r_P,r_Q$, respectively. 
Then 
\[PRQ\in \Aq^{\lhd(d+d_P+d_Q,\;r+r_P+r_Q)}.\]
\end{lemma}

\begin{proof}
By linearity, it is enough to consider the case where $R$ is a PBW basis element occurring in the definition of $\Aq^{\lhd(d,r)}$.

First suppose that $R$ has degree $e<d$.  Then every word occurring before
straightening in $PRQ$ has degree $e+d_P+d_Q<d+d_P+d_Q$.  By
Lemma~\ref{lem:straightening-respects-bidegree}, its PBW expansion contains
only terms of degree strictly less than $d+d_P+d_Q$.

Next suppose that $R$ has degree $d$ and length $s<r$.  Then every word
occurring before straightening in $PRQ$ has degree $d+d_P+d_Q$ and length
$s+r_P+r_Q<r+r_P+r_Q$.  By
Lemma~\ref{lem:straightening-respects-bidegree}, its PBW expansion contains
only terms of smaller degree, or of the same degree and length at most
$s+r_P+r_Q$.  These terms are again lower than
\[
        (d+d_P+d_Q,\;r+r_P+r_Q)
\]
in lexicographic order.
\end{proof}

\begin{proposition}
\label{prop:conjectural-spans}
For every $d\in\N$, one has
\[
        B_d\subseteq V_{\rm C}^{\leq d}.
\]
In particular, $\cB_{\rm C}$ spans $\Aq$.
\end{proposition}

\begin{proof}
It is enough to prove that every element of the known PBW basis
$\cB_{\rm T}$ lies in the appropriate filtered piece of $V_{\rm C}$.  We prove
the following stronger statement by induction on bidegree $(d,r)$ in
lexicographic order:
\[
        \text{if } w\in\cB_{\rm T} \text{ has degree } d
        \text{ and length } r,\text{ then } w\in V_{\rm C}^{\leq d}.
\]

Assume that the statement has been proved for all lower bidegrees, and let
$w\in\cB_{\rm T}$ have degree $d$ and length $r$.  With respect to the order
$<_{\rm T}$, the word $w$ has block form
\[
        \cG\text{-block}\;\cW^-\text{-block}\;
        \cW^+\text{-block}\;\widetilde{\cG}\text{-block}.
\]

We first move the $\cG$-letters to the right of the
$\cW^-$-letters.
Suppose that an adjacent subword is of the form
$\cG_{b+1}\cW_{-(N-b)}$.
By Corollary~\ref{cor:invert-crossings},
\[
        \cG_{b+1}\cW_{-(N-b)}
        =
        \sum_{a=0}^{N}c_{ba}^{(N)}
        \cW_{-(N-a)}\cG_{a+1}+R
\]
for some $R\in\Aq^{\lhd(2N+3,2)}$.  Multiplying this identity on the left and
right by the surrounding subwords, Lemma~\ref{lem:lower-stable} shows that
the contribution of $R$ lies in $\Aq^{\lhd(d,r)}$.  By the induction
hypothesis, this lower contribution belongs to $V_{\rm C}^{\leq d}$.

Thus, modulo $V_{\rm C}^{\leq d}$, we may replace each adjacent subword
of type $\cG\cW^-$ by a linear combination of subwords of type
$\cW^-\cG$.
Each such replacement decreases the number of inversions consisting of a
$\cG$-letter to the left of a $\cW^-$-letter.
Repeating finitely many times, and using the commutativity within each
family, we obtain
\[
        w\in V_{\rm C}^{\leq d}
        +
        \Span\{\,\cW^-\text{-block}\;\cG\text{-block}\;
        \cW^+\text{-block}\;\widetilde{\cG}\text{-block}\,\}.
\]

It remains to move the $\widetilde{\cG}$-letters to the left of the
$\cW^+$-letters.
Suppose that an adjacent subword is of the form
$\cW_{N-b+1}\Gt_{b+1}$.
Again by Corollary~\ref{cor:invert-crossings},
\[
        \cW_{N-b+1}\Gt_{b+1}
        =
        \sum_{a=0}^{N}d_{ba}^{(N)}
        \Gt_{a+1}\cW_{N-a+1}+R'
\]
for some $R'\in\Aq^{\lhd(2N+3,2)}$.
After multiplying by the surrounding subwords,
Lemma~\ref{lem:lower-stable} shows that the contribution of $R'$
lies in $\Aq^{\lhd(d,r)}$, and hence belongs to
$V_{\rm C}^{\leq d}$ by the induction hypothesis.

Therefore, modulo $V_{\rm C}^{\leq d}$, we may replace each adjacent
subword of type $\cW^+\widetilde{\cG}$ by a linear combination of
subwords of type $\widetilde{\cG}\cW^+$.
This decreases the number of inversions consisting of a
$\cW^+$-letter to the left of a $\widetilde{\cG}$-letter.
Repeating finitely many times, and using commutativity within each of the
four families, we obtain a linear combination of monomials in the block
order
\[
        \cW^-\text{-block}\;\cG\text{-block}\;
        \widetilde{\cG}\text{-block}\;\cW^+\text{-block}.
\]
These monomials belong to $\cB_{\rm C}$ and have degree $d$.  Hence
$w\in V_{\rm C}^{\leq d}$, completing the induction.

Since the elements of $\cB_{\rm T}$ of degree at most $d$ form a basis for
$B_d$, the inclusion $B_d\subseteq V_{\rm C}^{\leq d}$ follows.  Taking the
union over all $d$ shows that $\cB_{\rm C}$ spans $\Aq$.
\end{proof}

We now prove linear independence by comparing filtered dimensions.

\begin{proposition}
\label{prop:conjectural-basis}
The alternating generators of $\Aq$, ordered by increasing index within
each family and by the block order
\[
        \cW^-<\cG<\widetilde{\cG}<\cW^+,
\]
give a PBW basis for $\Aq$.
Equivalently, the set $\cB_{\rm C}$ is a basis for $\Aq$.
\end{proposition}

\begin{proof}
The spanning assertion follows from Proposition~\ref{prop:conjectural-spans}.
It remains to prove linear independence.

Fix $d\in\N$.  Let $P_d$ be the vector space with basis indexed by the formal
monomials of the form \eqref{eq:conjectural-basis-monomial} of degree at most
$d$.  Multiplication in $\Aq$ gives a linear map
\[
        \mu_d:P_d\longrightarrow B_d
\]
which sends each formal monomial to the corresponding product in $\Aq$.  This
map is well-defined because an $\Aq$-word of degree at most $d$ lies in $B_d$.
By Proposition~\ref{prop:conjectural-spans}, the map $\mu_d$ is surjective.

It remains to compare dimensions.  The degree generating function for the
formal monomials in the conjectural block order is
\[
        \prod_{k\geq 0}
        \frac{1}{(1-x^{2k+1})^2}
        \frac{1}{(1-x^{2k+2})^2}
        =
        \prod_{n=1}^{\infty}\frac{1}{(1-x^n)^2}.
\]
By \eqref{eq:hilbert}, this is the Hilbert series of $\Aq$ with respect to
the degree filtration.  Therefore, after summing coefficients up to degree
$d$, one obtains
\[
        \dim P_d=\dim B_d.
\]
Since $\mu_d$ is a surjective linear map between vector spaces of the same
finite dimension, it is an isomorphism.  Thus the images in $\Aq$ of the
formal monomial basis of $P_d$ are linearly independent.

As $d$ was arbitrary, all monomials in $\cB_{\rm C}$ are linearly independent.
Together with spanning, this proves that $\cB_{\rm C}$ is a basis for $\Aq$.
\end{proof}

\begin{corollary}
\label{cor:four-orders}
Keeping the increasing-index order within each family, the alternating
generators give PBW bases for each of the following four block orders:
\[
        \cG<\cW^-<\cW^+<\widetilde{\cG},
        \qquad
        \cW^-<\cG<\cW^+<\widetilde{\cG},
\]
\[
        \cG<\cW^-<\widetilde{\cG}<\cW^+,
        \qquad
        \cW^-<\cG<\widetilde{\cG}<\cW^+.
\]
\end{corollary}

\begin{proof}
The first block order is given by
Theorem~\ref{thm:terwilliger-pbw}.
The fourth block order is given by
Proposition~\ref{prop:conjectural-basis}.
The second is obtained by the same argument as
Proposition~\ref{prop:conjectural-basis}, using only the invertibility of
the crossing $\cG\cW^-\leftrightarrow\cW^-\cG$.
The third is obtained similarly, using only the invertibility of the crossing
$\cW^+\widetilde{\cG}\leftrightarrow\widetilde{\cG}\cW^+$.
\end{proof}

\section{Proof of Terwilliger's conjecture}

We conclude by translating the PBW basis obtained above into the tensor
factorisation in Terwilliger's conjecture
\cite[Conjecture~4.5]{Terwilliger2018Action}.

\begin{proof}[Proof of Theorem~\ref{thm:intro-main}]
Parts~{\rm (i)--(iv)} follow from Proposition~\ref{prop:poly}.

For part~{\rm (v)}, Proposition~\ref{prop:conjectural-basis} states that the
monomials
\[
        \cW_{-i_1}\cdots \cW_{-i_r}\,
        \cG_{j_1+1}\cdots \cG_{j_s+1}\,
        \Gt_{k_1+1}\cdots \Gt_{k_t+1}\,
        \cW_{\ell_1+1}\cdots \cW_{\ell_u+1},
\]
with non-decreasing indices in each of the four blocks, form a basis for
$\Aq$.
By Proposition~\ref{prop:poly}, the standard monomial bases of
$\cW^-$, $\cG$, $\widetilde{\cG}$, and $\cW^+$ are given by the
corresponding single-family ordered monomials.
Hence the tensor-product monomial basis of
\[
        \cW^-\otimes\cG\otimes
        \widetilde{\cG}\otimes\cW^+
\]
is mapped by multiplication bijectively onto the basis
$\cB_{\rm C}$ of $\Aq$.

Therefore, the multiplication map
\[
        \cW^-\otimes\cG\otimes
        \widetilde{\cG}\otimes\cW^+
        \longrightarrow\Aq,
        \qquad
        u\otimes v\otimes w\otimes x\longmapsto uvwx
\]
is an isomorphism of vector spaces.
\end{proof}

	\bigskip
	\centerline{\bf Acknowledgements}
The author thanks the referees for their careful reading and helpful suggestions, which have improved the motivation, organisation, and presentation of the paper. H.Z. gratefully acknowledges support from the NTU Research Scholarship.

\end{document}